\documentclass{article}
\usepackage{ijcai26}

\usepackage{times}
\usepackage{soul}
\usepackage{url}
\usepackage[hidelinks]{hyperref}
\usepackage[utf8]{inputenc}
\usepackage[small]{caption}
\usepackage{graphicx}
\usepackage{amsmath}
\usepackage{amsthm}
\usepackage{booktabs}
\usepackage{algorithm}
\usepackage{algorithmic}
\usepackage[switch]{lineno}

\usepackage{times}
\usepackage{soul}
\usepackage{url}
\usepackage[hidelinks]{hyperref}
\usepackage[utf8]{inputenc}
\usepackage{amsmath}
\usepackage{amsthm}
\usepackage{booktabs}
\usepackage[switch]{lineno}
\usepackage[dvipsnames]{xcolor}
\usepackage{xspace}

\usepackage{hyperref}
\usepackage{url} 
\hypersetup{
  breaklinks=true,
}

\usepackage{amssymb}
\usepackage{tikz}
\usetikzlibrary{positioning}
\usetikzlibrary{chains,fit,shapes}

\tikzset{%
round/.style={circle, draw=black,fill=gray!5, very thick,minimum size=10mm,},
dot/.style={draw, circle, minimum size=2mm,inner sep=0pt,outer sep=0pt,fill=black},% and so on
}

\newtheorem{theorem}{Theorem}
\newtheorem{definition}{Definition}
\newtheorem{lemma}{Lemma}

\newcommand{\cA}{\mathcal{A}}
\newcommand{\cM}{\mathcal{M}}
\newcommand{\cD}{\mathcal{D}}

\newcommand{\cF}{\mathcal{F}}
\newcommand{\cS}{\mathcal{S}}

\newcommand{\cV}{\mathcal{V}}
\newcommand{\cR}{\mathcal{R}}

\newcommand{\cC}{\mathcal{C}}

\newcommand{\cT}{\mathcal{T}}

\newcommand{\cI}{\mathcal{I}}
\newcommand{\cO}{\mathcal{O}}

\newcommand{\cJ}{\mathcal{J}}

\newcommand{\dtime}{\mathsf{DTIME}}
\newcommand{\ntime}{\mathsf{NTIME}}

\newcommand{\leftmarker}{\triangleright}
\newcommand{\rightmarker}{\triangleleft}

\newcommand{\ConfCM}{\ensuremath{\mathit{Conf}^{CM}}\xspace}
\newcommand{\ConfLBA}{\ensuremath{\mathit{Conf}^{LBA}}\xspace}
\newcommand{\ConfTM}{\ensuremath{\mathit{Conf}^{TM}}\xspace}
\newcommand{\deltaFinal}{\ensuremath{\delta^{F}}}

\newtheorem*{definition*}{Definition}
\usepackage{thmtools}
\usepackage{thm-restate}

\usepackage{todonotes}
\newcommand\mg[1]{\todo[color=green!30,size=\small,inline]{\textbf{Maksim}: #1}}

\newcommand\natasha[1]{\todo[color=cyan!30,size=\small,inline]{\textbf{Natasha}: #1}}
\newcommand\bsl[1]{\todo[color=red!30,size=\small,inline]{\textbf{Brian}: #1}}

\title{Temporal Causal Models as a Model of Computation}

\author{%
Maksim Gladyshev$^1$\and
Natasha Alechina$^{2, 1}$\and
Brian Logan$^{3, 1}$\\
\affiliations
$^1$Utrecht University, The Netherlands\\
$^2$Open Universiteit, The Netherlands\\
$^3$University of Aberdeen, UK\\
\emails
\{m.gladyshev, n.a.alechina, b.s.logan\}@uu.nl
}

\begin{document}

\maketitle

\begin{abstract}
Causal models, also known as Structural Equation Models (SEM), are a well-known formalism for representing and reasoning about causal dependencies between events. In this paper, we show that Temporal SEMs (TSEMs), which extend SEMs to support causal reasoning in temporal settings, can be interpreted as a model of computation. We prove that TSEMs can encode Linear Bounded Automata, and thus causal settings representable in context sensitive languages. We also prove that TSEMs with countably many variables are Turing complete. These results establish a formal connection between causal reasoning and classical models of computation, enabling the integration of counterfactual reasoning techniques from causal inference into the theory of computation.
% Causal models, also known as Structural Equation Models (SEM), are a well-known formalism for representing and reasoning about causal dependencies between events. In this paper, we show that Temporal SEM (TSEM), which allow reasoning about causality in temporal settings, can be used as a model of computation. We prove that TSEMs can encode Linear Bounded Automata and thus causal settings representable in context sensitive languages. We also prove that TSEMS with countably many variables are Turing complete, and sufficiently expressive for reasoning about computations in both deterministic and non-deterministic settings.  %bsl: The following needs revised
% These results allow us to enrich the Theory of Computation with counterfactual reasoning techniques developed in the field of causal reasoning.

\end{abstract}

\section{Introduction}
%bsl: replaced with the introduction from the abandoned AAAI 2026 paper.

% Possibly relevant papers: \cite{Scheines2002} \cite{Collier2011} \cite{Cooper_Glymour_1999} 

Causal models in the form of Structural Equation Models (SEM) \cite{Pearl2000,Spirtes2001,HalpernBook} have become the dominant AI framework for representing and analysing causal relations between events of interest. These models represent events as the assignment of values to variables, and encode causal dependencies between variables using  structural equations which define how the value of each variable depends on the values of other variables. Being a relatively simple yet powerful formalism, they have found multiple applications in different disciplines, including Economics \cite{NBERw29787}, Medicine \cite{Gunzler_book} and Machine Learning \cite{Albahri2022}. 

Causal frameworks can also be used to causally model and understand the temporally extended behaviour of natural systems or temporally extended processes. In many applications we are interested in reasoning about not only whether ``event A causes event B" but also whether ``event A happening at time $t_1$ causes event B at time $t_2$", and a range of approaches to integrating temporal and causal reasoning have been proposed, based both on SEMs  e.g., \cite{bongers2022causalmodelingdynamicalsystems,Assaad_2022,boeken2024dynamicstructuralcausalmodels,Gladyshev_Alechina_Dastani_Doder_Logan_2025,Cinquini_Beretta_Ruggieri_Valera_2025}, and on alternative frameworks, e.g., \cite{CONSOLE19931,Lorini01102013,Batusov_Soutchanski_2018,Ciuni2018,Finkbeiner2024}.

In this paper we argue that temporal causal models are not only an effective tool for reasoning about causal structures underlying physical, biological or social processes, but can also be viewed as a useful formalism for representing computing devices. Although other temporal formalisms, for example extensions of the logic LTL \cite{mukherjee2012synchronizing}, can be used to reason about computations, temporal causal models allow us to use techniques specifically developed for causal inference, e.g., counterfactual reasoning, to model and understand the dynamics of program behaviour. 
%bsl: revised
We show that Temporal Structural Equation Models (TSEM) \cite{Gladyshev_Alechina_Dastani_Doder_Logan_2025} can encode Linear Bounded Automata and thus causal settings representable in context sensitive languages. 
We also prove that TSEMS with countably many variables are Turing complete, and sufficiently expressive for reasoning about computations in both deterministic and non-deterministic settings.
%
% The implications of this are twofold. 

% First, it demonstrates that  TSEMs are as expressive as Turing machines, in the sense that all behaviours representable by a Turing machine can be represented by a TSEM. Second, it shows that they can be used as theoretical models of computation. 
%and, thus, can be used as theoretical models of computation. 

Understanding the causal evolution of computational processes over time is important for computability theory, and we argue that TSEMs provide a useful framework to represent not only the sequence of computational states, but also the causal dependencies between state transitions. This allows us to analyse how specific inputs or intermediate steps causally affect computational outcomes. For example, TSEMs may be used to augment formal verification techniques to determine why certain errors occurred in a system, and what could prevent them. 

The contributions of this paper are as follows:
\begin{enumerate}
    \item We introduce a non-deterministic extension of TSEMs and define %two types of 
    interventions on TSEMs in Section \ref{sec:temporalmodels}.
    \item We prove that TSEMs are as expressive as Linear Bounded Automata in Section \ref{sec:LBA}.
    \item We prove that TSEMs with countably many variables are as expressive as Turing machines in Section \ref{sec:TM}.
    \item We discuss how interventions on TSEMs can be used to answer computational questions in Section \ref{sec:interventions}.
\end{enumerate}

% The paper is organized as follows. In Section \ref{sec:temporalmodels} we introduce temporal causal models framework. In Section \ref{sec:LBA} we present our main result proving that TSEMs can simulate LBAs. In Section \ref{sec:TM} we extend this result proving that TSEMs with countably many variables are Turing complete. Then, we define several types of causal interventions suitable for the considered temporal settings in Section \ref{sec:interventions}. We advocate the use of TSEMs as a model of computation and argue that enriching the Theory of Computation with counterfactual reasoning techniques  by means of causal interventions has multiple important applications. Finally, we conclude and discuss potential directions for future work in Section \ref{sec:conclusion}. 

% \mg{Maybe to omit either contributions or the paper structure to save space?}

%\maksim{the paper to mention somewhere: \cite{Asparouhov04052018}}

% !TEX root = ijcai26.tex

\section{Temporal Causal Models} 
\label{sec:temporalmodels}

%bsl: revised for non-determinism and readability; needs checked for correctness
In this section, we introduce the main components of Temporal Structural Equation Models (TSEM). As we aim to encode nondeterministic Linear Bounded Automata and nondeterministic Turing Machines, we focus on \emph{nondeterministic} causal models in which a set of events may result in more than one outcome.
%The presentation follows that in \cite{Gladyshev_Alechina_Dastani_Doder_Logan_2025}. 
%MG: "a set of events may result in more than one outcome" is a bit unclear to me, I'd reformulate

Let $\cV = \{X_1, X_2, \dots\}$ be a finite set of variables. 
%MG: to extend the set to countable in Sec. 4.
%\footnote{In Section \ref{sec:countable } we consider TSEMs over a countable $\cV$.} 
A \emph{range function} $\cR$  associates with every variable $X_i\in \mathcal{V}$ a non-empty \textit{finite} set $\mathcal{R}(X_i)=\{x_1, \dots, x_k\}$ of possible values. The values of variables may change over time, and these changes are governed by \emph{temporal structural equations} which describe how the current value of each $X_i\in \cV$ depends on the previous values of variables in $\cV$. A \textit{domain function} $\cD: \cV \longrightarrow 2^{\cV}$ maps each variable $ Y \in \cV$ to a \emph{finite} subset of variables, intuitively those variables whose values determine the possible next values of $Y$.
%MG: add a remark about 1-step equations in Def.1

A causal model $\cM$ over a signature $\cS= (\cV, \cR, \cD)$ consists of a set of \emph{structural equations} $\cF = \{\cF_Y\mid Y\in \cV\}$ specifying how the next values of $Y$ depends on the previous values of the variables in its domain $\cD(Y)$. Formally:

\begin{definition}[Temporal Structural Equation Model]\label{def:SEM} A (nondeterministic) Temporal Structural Equation Model (TSEM), over a signature $\mathcal{S}$ is a tuple $\cM = (\mathcal{S}, \cF)$, where $\cF$ associates with every variable $Y \in \mathcal{V}$ a function 
$$\mathcal{F}_Y:\prod\limits_{X \in \cD(Y)}\cR(X)\longrightarrow 2^{\cR(Y)}$$
which defines the structural equation describing how the (next) value of $Y$ depends on the (previous) values of the variables in its domain $\cD(Y)$.\footnote{We abuse notation and interpret elements of $\prod\limits_{X\in \{X_1, \dots, X_k\}}\cR(X)$ both as vectors of values $(x_1,\dots,x_k)$ and as vectors of assignments $(X_1=x_1,\dots,X_k=x_k)$.} A TSEM is deterministic if all structural equations always return a single value. 
\end{definition}

%bsl: revised for readability
The changes in the values of variables induced by a set of structural equations are described by a computation. A \emph{computation} $\cC$ is a tree  of \emph{configurations} of $\cM$, i.e., complete assignments of values to variables $\cV$. We denote a complete assignment to $\cV$ by $\vec{v}$, and use $\vec{v}^{|X}$ to denote $\vec{v}$ restricted to $X \in \cV$; we also denote $\vec{v}$ restricted to the variables in the domain of $Y$ as  $\vec{v}^{|\cD(Y)}$. Let $\ConfCM(\cM)$ be the set of all possible configurations of a causal model $\cM$. 

%bsl: linear computations ...
% \begin{definition}[Computation]\label{def:computation} Given a TSEM $\cM$ and an initial configuration $\vec{v_0}$, a computation $\cC$ is a function 
% $$\cC: \mathbb{N}\longrightarrow \prod\limits_{X\in \cV} \cR(X)$$
% mapping $\mathbb{N}$ to configurations of $\cV$, such that $\cC(0)=\vec{v_0}$ and for all $i>0$: 
% $$\cC(i) = \prod\limits_{X\in \cV} \cF_X\bigl(\cC(i-1)^{|\cD(X)}\bigr)$$
% where by $\cC^{|\cD(X)}(i-1)$ we mean the configuration $\cC(i-1)$ restricted to the variables in the domain of $X$. 
% \end{definition}

% Intuitively, a computation $\cC$ proceeds from a given initial configuration $\vec{v_0}$ and consecutively `updates' the values of $\cV$ according to $\cF$. If $\cC(i)=\cC(i+1)$ for some $i$, then $\cC(i)=\cC(i+j)$ for any $j>0$, and we say that $\cC$ \emph{terminates} after $i$ steps. 

% Then $\cM$ can be represented as a set of binary relations $\twoheadrightarrow_{\cM}\subseteq Conf^{CM}(\cM)\times Conf^{CM}(\cM)$, where $\vec{v_1}\twoheadrightarrow_{\cM}\vec{v_2}$
% %$\vec{v_1}, \vec{v_2}\in Conf^{CM}(\cM)$ 
% means that $\vec{v_2} = \prod\limits_{X\in \cV} \cF_X\bigl(\vec{v_1}^{|\cD(X)}\bigr)$, i.e., configuration $\vec{v_2}$ is a possible outcome of configuration $\vec{v_1}$ in a single step.
% \mg{to change: $\vec{v_2} = \prod\limits_{X\in \cV} \cF_X\bigl(\vec{v_1}^{|\cD(X)}\bigr)$ above, or just remove since Def. 2 explains this again}
%where by $\vec{v_1}^{|\cD(X)}$ we mean the configuration $\vec{v_1}$ restricted to the variables in the domain of $X$. 
%Formally:

\begin{definition} [Computation]\label{def:computation}
A computation $\cC$ for a TSEM $\cM$ with initial assignment $\vec{v_0}$ is a tree $\cT^{\cC}$, whose nodes are elements of $\ConfCM(\cM)$ with root node $\vec{v_0}$, and $\vec{v_1} \twoheadrightarrow_{\cM} \vec{v_2}$ iff\ $\forall X\in\cV, \vec{v_2}^{|X}\in \cF_X(\vec{v_1}^{|\cD(X)})$.
\end{definition}

%As the simplest illustration, consider the deterministic model $\cM_1$ which contains a single variable $X$ with a binary range $\cR(X) \mathbin{=} \{0, 1\}$, such that $\cF_X(X\mathbin{=}0) \mathrel{:=} 1$ and $\cF_X(X=1):=0$. Thus, for an initial configuration $\vec{v} = (X=0)$ it generates a computation $\cC_1$, in which $X$ repeatedly `flips' its value: $\cC_1 = (X=0, X=1, X=0, \dots)$ and never terminates. 

As a simple illustration, consider a model $\cM$ with a single variable $X$: $\cV=\{X\}$, such that $\cR(X)= \{0, \dots, 9\}$, $\cD(X) = \{ X \}$, and a single structural equation defined as follows: 
$$\cF_X (X=x)=\begin{cases}
    9, \text{ if } x=9\\
    \{0, x+1\}, \text{ if } x\neq 9\\
\end{cases}$$ 
If the value of $X$ on the previous step of the computation is equal to 9, it remains so on the current step. However, if the previous value $x$ of $X$ is less than 9, $\cF_X$ makes a non-deterministic choice: either to continue counting and increase the value of $X$ to $x+1$ or to reset $X$ to 0.

One of the key features of causal models is that they allow \emph{counterfactual reasoning}, that is, they allow us to reason not only about the %cause(s) of a particular event
actual sequence of events, but also about hypothetical counterfactual sequences of events by means of so-called \emph{interventions}. An intervention is an operation on causal model.
%
%We distinguish two types of interventions: \emph{value} interventions and \emph{structure} interventions. Value interventions set the values of some variables during the computation, but do not affect structural equations, while structure interventions change the structural equations, but do not necessarily set variables' values. 

An atomic
%atomic \emph{value} 
intervention is denoted by $do(Y^n\gets y)$, where $Y\in\cV, n\in\mathbb{N}, y\in\cR(Y)$, says that the value of the variable $Y$ must be fixed to $y$ at the $n$th step of the computation.
An 
%A value 
intervention 
$do(\vec{Y}^{\vec{n}}\leftarrow \vec{y}) = do(Y_1^{n_1}\gets y_1, \dots, Y_i^{n_j}\gets y_i)$ is a vector of atomic %value
interventions. Here we assume that if $Y_k^n\gets y_k$ and $Y_l^n\gets y_l$ occur in $\vec{Y}^{\vec{n}}\leftarrow \vec{y}$, then $Y_k\neq Y_l$. This intervention results in an updated computation tree $\cC^{\vec{Y}^{\vec{n}}\leftarrow \vec{y}}$, defined for each computation branch as follows. Given a default assignment $\vec{v}_0$, let $\vec{v}^{I}_0$ under intervention $I=do(\vec{Y}^{\vec{n}}\leftarrow \vec{y})$ be an assignment of $\cV$, which agrees with $\vec{v}_0$ everywhere, except the variables $Y_i$, such that $Y_i^0\leftarrow y$ occurs in  $\vec{Y}^{\vec{n}}\leftarrow \vec{y}$. The values of those variables in $\vec{v}^{I}_0$ are set according to $\vec{Y}^{\vec{n}}\leftarrow \vec{y}$. We write $\cC(i)|_X$ to refer to $X$'s value at the $i$'th step of $\cC$.
This means that we set a value of $X$ on all branches of $\cC$.

%\mg{The interventions are defined for deterministic computation now}

%\begin{definition}[Value Intervention] 
\begin{definition}[Intervention]\label{def:intervention}
    Given an intervention $I=do(\vec{Y}^{\vec{n}}\leftarrow \vec{y})$ and $(\cM, \vec{v_0})$, an \emph{updated} computation tree $\cC^{I}$ is defined as, for every branch $\cC_b$ in the tree, $\cC^{I}_b(0)=\vec{v_0}^{I}$, and $\forall i>0,\, \forall X\in \cV$: 
$$\cC^{I}_b(i)|_X=\begin{cases}
			x', \text{ if } X^i\gets x' \in \vec{Y}^{\vec{n}}\leftarrow \vec{y}\\
            \cF_X(\cC^{I}_b(i-1))^{|\cD(X)}, \text{ otherwise}
		 \end{cases}$$
\end{definition}

Interventions can be used to define a \emph{cause} of a particular state of affairs. Since we consider a temporal setting, as in \cite{Gladyshev_Alechina_Dastani_Doder_Logan_2025}, we focus on the notion of a \emph{cause at step $n$} rather than cause in general. 
%We follow the definitions in \cite{Gladyshev_Alechina_Dastani_Doder_Logan_2025} to define a cause. 
Since we are 
in non-deterministic setting, we follow \cite{beckers2025nondeterministic} to define a cause in the strong sense that the outcome could have been prevented on all branches of the computation tree.

\begin{definition}\label{def:cause}
Given a TSEM $\cM$, its initial configuration $\vec{v}_0$ and two assignments of variables $\vec{X}^{\vec{t}}=\vec{x} = (X_1^{t_1}=x_1, \dots, X_n^{t_n}=x_n)$ and $\vec{Y}^{\vec{t'}}=\vec{y} = (Y_1^{t_1'}=y_1, \dots, Y_k^{t_k'}=y_k')$ at the corresponding time steps $t_1, \dots, t_n$ and $t_1', \dots, t_k'$. We say that  $\vec{X}^{\vec{t}}=\vec{x}$ is the \emph{cause} of $\vec{Y}^{\vec{t'}}=\vec{y}$ in $(\cM, \vec{v}_0)$ iff \\ 
(1) both $\vec{X}^{\vec{t}}=\vec{x}$ and $\vec{Y}^{\vec{t'}}=\vec{y}$ hold in $\cC$ for $(\cM, \vec{v}_0)$;\\
(2) there exists an intervention $I = do(\vec{X}^{\vec{t}}\gets \vec{x}')$ with $\vec{x}\neq \vec{x'}$, such that $\vec{Y}^{\vec{t'}}=\vec{y}$ does not hold on any branch in $\cC^I$;\\
(3) no proper subset of $\vec{X}^{\vec{t}}=\vec{x}$ satisfies condition (2). 
\end{definition}

The definition of causal models we present here differs from the standard definition of static SEM, e.g., \cite{HalpernBook}, where a signature usually contains two types of variables: exogenous variables (whose values are determined outside of the model) and endogenous variables (whose values depend on the values of other variables). However, this distinction is not required in our setting, and we treat all the variables as endogenous. 

%\bsl{Add a discussion of non-standard interventions.}

Following \cite{Gladyshev_Alechina_Dastani_Doder_Logan_2025}, we also interpret structural equations as temporal (or time-lagged), and call these models \textit{temporal} SEMs. While in general TSEMs allow equations with arbitrary long delays (time lags) \cite{gladyshev2025temporal}, to simplify notation we restrict consideration to TSEMs with 1-step delays as these are sufficient for the proofs presented below. Lastly, while  non-deterministic extensions of static SEMs have been considered in \cite{beckers2025nondeterministic,beckers2025actual} and \cite{Barbero2024}, to the best of our knowledge, temporal extensions of non-deterministic SEMs have not been studied before.

% !TEX root = ijcai26.tex

\section{Encoding Linear Bounded Automata}\label{sec:LBA}

In this section, we show how TSEMs can be used to encode Linear Bounded Automata.

A \emph{linear bounded automaton} (LBA) is a restricted form of a Turing Machine in which the tape is finite and bounded by left and right endmarkers, and where transitions may not overwrite the endmarkers or to move to the left of the left endmarker or to the right of the right endmarker. 
%bsl: we define deterministic LBAs below.
Non-deterministic LBAs accept the class of context-sensitive languages, that is, problems which can be expressed in a context sensitive grammar \cite{Hopcroft/Ullman:79a}. An LBA accepts a grammar if for each string in the language it accepts with a tape that is bounded by a linear function of the length of the input string. 
%the string's length multiplied by a constant. 

% \begin{definition}[LBA]\label{def:comp:tm} An LBA is a tuple $T = (Q, q_0, F, \Gamma, \Sigma, \#, \delta)$, where
% \begin{itemize}
%     \item $Q$ is a finite set of states,
%      \item $q_0\in Q$ is the initial state,
%     \item $F\subset Q$ is the set of finial (or accepting) states,
%      \item $\Gamma$ is a finite set of tape symbols,
%     \item $\Sigma\subseteq\Gamma$ is a finite set of input symbols,
%     \item $\#\in \Gamma$ is the blank symbol, such that $\#\notin \Sigma$, i.e. $\#$ is not an input symbol, 
%     \item $\leftmarker \in \Gamma$ and $\rightmarker \in \Gamma$
%     are left endmarker and right endmarker respectively, such that
%     $\leftmarker \notin \Sigma$ and $\rightmarker \notin \Sigma$,
%     \item $\delta: (Q\setminus F)\times \Gamma \longrightarrow Q\times\Gamma \times \{Left, Right\}$ is a function mapping a state $q\in (Q\setminus F)$ and a symbol $\gamma\in \Gamma$ to a tuple $\delta(q, \gamma) = (q', \alpha, D)$, where $q'$ is the next state, $\alpha\in \Gamma$ is a symbol that must be written in the cell and $D\in \{Left, Right\}$ is the direction in which the head moves
%     \item constraint on $\delta$ \natasha{will cite from Hopcroft/Ullman, but for now:} if $\gamma$ is $\leftmarker$ then $\alpha=\leftmarker$ and $D \not = Left$ and if $\gamma$ is $\rightmarker$ then $\alpha=\rightmarker$ and $D \not = Right$.
% \end{itemize}
% \end{definition}

There are several definitions of LBAs in the literature. The definition below is based on that in \cite{Simovici/Tenney:99a} as it allows a simpler encoding as a causal model. 
\begin{definition}[LBA]\label{def:comp:tm} A (nondeterministic) LBA is a tuple $\cA = (Q, q_0, F, \Sigma, \Gamma, \delta)$, where
\begin{itemize}
    \item $Q$ is a finite set of states with $q_0\in Q$ the initial state;
    \item $F\subset Q$ is the set of final (or accepting) states;
    \item $\Sigma$ is a finite set of input symbols;
    \item $\Gamma = \Sigma \cup \{\#, \leftmarker, \rightmarker \}$ is the finite set of tape symbols, where $\#$ is the blank symbol, $\leftmarker$ and $\rightmarker$ the left and right endmarkers, and $\{\#, \leftmarker, \rightmarker\} \cap \Sigma = \emptyset$;
    \item $\delta: (Q\setminus F)\times \Gamma \times Q\times\Gamma \times \{-1, 0, 1\}$ is a transition relation mapping a state $q\in (Q\setminus F)$ and a symbol $\gamma\in \Gamma$ to a set of tuples     $(q', \gamma', d)$, where $q'$ is the next state, $\gamma' \in \Gamma$ is the symbol written in the current cell and $d\in \{-1,0, 1 \}$ is the direction in which the head moves ($-1$ indicates the head moves to the left etc.). 
%bsl: redefined to allow transitions in final states
%    \item $\delta: Q \times \Gamma \times Q\times\Gamma \times \{-1, 0, 1\}$ is a transition relation mapping a state $q\in Q$ and a symbol $\gamma\in \Gamma$ to a set of tuples $(q', \gamma', x)$, where $q'$ is the next state, $\gamma' \in \Gamma$ is the symbol written in the current cell and $d\in \{-1,0, 1 \}$ is the direction in which the head moves ($-1$ indicates the head moves to the left etc.). 
\end{itemize}
\end{definition}
\noindent
The transition relation $\delta$ is constrained such that, if $\gamma$ is $\leftmarker$ then $\gamma' = \leftmarker$ and $d \not = -1$ and if $\gamma$ is $\rightmarker$ then $\gamma' = \rightmarker$ and $d \not = 1$, i.e., the endmarkers cannot be overwritten and transitions cannot move to the left of the left endmarker or to the right of the right endmarker. 
%bsl: do we also need to stipluate that no transition can write an endmarker symbol bewtween the endmarkers?
Note also that $\delta$ is defined only for states in $Q \setminus F$, i.e., computation stops when a final state is reached.

We say that an LBA has tape of length $n$ if the number of cells between $\leftmarker$ and $\rightmarker$, not including the cells with the endmarkers, is $n$.
%To simplify notation, we will write %$\mathit{L}=  
%$Dir(q, \gamma)$  for the set of possible values for a direction, %$\alpha= 
%$Sym(q, \gamma)$  for the set of possible values for a symbol) and %$q' = 
%$Out(q, \gamma)$  for the set of possible values for the next state.
%
We represent a \emph{configuration} of an LBA with a tape of length $n$ as a string 
$$\lambda = C_{0} C_{1}  \ldots q C_i \ldots C_{n} C_{n + 1}$$ 
where $q$ is the state of the LBA, the tape head is above cell $C_i$, $C_{0}=\leftmarker$, $C_{n + 1}=\rightmarker$, and each $C_k, 1 \leq k \leq n$ represents a symbol written in the $k$-th cell of the tape.
%bsl: needs checked
Without loss of generality, we assume that, in the initial configuration, the head is above cell $C_0$, and the string of input symbols $s = \sigma_1, \ldots, \sigma_k$ are in the cells $C_1, \ldots C_{k}$. 

\begin{definition} Let $\ConfLBA({\cal A})$ be the set of all possible configurations of an LBA ${\cal A}$. Then, a run of ${\cal A}$ on input $s$ is a tree $\cT$. The nodes of $\cT$ are elements of $\ConfLBA({\cal A})$ with a root node $\lambda_s$ representing the initial configuration of $\cal A$ for input $s$, and $\lambda \twoheadrightarrow_{\cal A}\lambda' \in \cT$ iff $\lambda'$ is a possible result of transition from $\lambda$ according to $\delta$.
\end{definition}
An LBA \emph{accepts} a string $s$ iff the computation tree $\cT$ for the input $s$ contains a sequence of valid transitions $\lambda_s \twoheadrightarrow_{\cal A} \dots \twoheadrightarrow_{\cal A}\lambda'$ from the root node $\lambda_s$ to some accepting configuration $\lambda' = C_{0} \ldots q^F \ldots C_{n + 1}$ where the LBA is in a final state  (i.e., $q^F \in F$).
%A configuration $\lambda' = C_{0} C_{1} \ldots C_{n} q^F C_{n + 1}$ is accepting if the LBA is in a final state  (i.e., $q^F \in F$) and the head is above the right endmarker. 

We now show how to construct a TSEM which we call a \emph{Causal Calculator} that simulates an LBA. 
%bsl: we should avoid this assumption if we can.
%We consider a special type of TSEM called a \emph{Causal Calculator} where the structural equations $\cF$ may be partially defined, i.e., some branches of the computation $\cC$ can terminate to match the behavior of a terminating computation.
%
For a given LBA $\cA$ with a tape of length $n$ we are going to build an LBA causal calculator $\cM^{\cA}$ parameterised by $n$ such that $\cA$ accepts a string iff $\cM^{\cA}$ has an accepting  computation on the corresponding input. 
To simplify the construction, we introduce a transition relation $\deltaFinal$ which is deterministic for final states of the LBA:
$$\deltaFinal(q, \gamma) = \begin{cases} 
    \{ \langle q', \gamma', d \rangle \mid \langle q, \gamma, q', \gamma', d \rangle \in \delta(q, \gamma) \} \text{ if } q \not\in F \\
    \{ \langle q, \gamma, 0 \rangle \}  \text{ if } q\in F 
    \end{cases}$$
% $$\deltaFinal(q, \gamma) = \begin{cases} 
%     \{ \langle q', \gamma', d \rangle \mid \\
%     \quad \langle q, \gamma, q', \gamma', d \rangle \in \delta(q, \gamma) \} \text{ if } q \not\in F \vee \gamma \not = \rightmarker \\
%     \{ \langle q, \gamma, 0 \rangle \}  \text{ if } q\in F \wedge \gamma = \rightmarker
%     \end{cases}$$
where $d \in \{-1, 0, 1\}$.

We first explain the intuition and then give the formal definition.
As the LBA transition relation $\delta$ is non-deterministic, we cannot update the state, symbol and the direction in which the head moves separately in the configuration of $\cM^{\cA}$. 
%The value of the variable $X_0$ is a 3-tuple which represents the result of the last transition of the LBA, i.e., a state of the LBA, the symbol written in the cell currently under the head, and the direction in which the head will move next $\{-1, 0, 1\}$.
We therefore use a special variable $X_0$ whose value is a 3-tuple where the first element is the current state of the LBA, the second is the symbol written at the last LBA transition, and the third is the direction the head moved to arrive at the current LBA configuration. In addition, to allow interventions on LBA transitions, instead of an encoding where a variable $X_i$ holds the symbol in the $i$th cell of the LBA, the index of the variable holding the symbol in $C_i$ is relative to the last position of the head. In effect, we shift the representation of the LBA tape in the causal model left or right depending on the movement of the head. This requires twice the number of variables in the causal model as there are cells in the tape of the LBA, from $X_{-(n+1)}$ to $X_{n+1}$. 
%We do this to be able to intervene on the current state of the LBA, and because this encoding can also be used for infinite tape Turing machines. However the price to pay for this ability is having twice more variables in the causal model than there are cells on the LBA tape: since we shift the tape left and right around the head, we need variables from $X_{-(n+1)}$ to $X_{n+1}$. Those $X_i$ where $i \not = 0$ hold the symbol in the LBA tape cell that corresponds to $X_i$; but because of the head movement, the cell may have an index different from $i$. The idea is that there is 
At any given time, a contiguous set of these variables correspond to $C_0,\ldots,C_{n + 1}$ but shifted by at most $n+1$ compared to the LBA cells, depending on the past movements of the head. For example, if at time $t - 1$ the head was over $C_{n+1}$, at time $t$ the second element of $X_0$ is $\rightmarker$ and $C_0$ in the LBA corresponds to the variable $X_{-(n+1)}$.

%bsl: The value of $X_0$ represents either the initial state or the result of the last transition.

\begin{definition}\label{def:CCLBA} Given an LBA $\cA = (Q, q_0, F, \Sigma, \Gamma, \delta)$ with tape of length $n$,  an \emph{LBA Causal Calculator} is a tuple $\cM^{{\cal A}} = (\cS^{\cal A}, \cF^{\cal A})$, with signature $\cS^{\cal A} = (\cV, \cR, \cD)$, where
\begin{itemize}
    \item $\cV = \{X_i\mid -(n + 1) \leq i \leq (n + 1)\}$;
    \item $\cR(X_0) = Q \times \Gamma \times \{-1, 0, 1\}$, \\ 
    $\cR(X_{i \in [-(n + 1), n + 1] \setminus \{0\}}) = \Gamma$;
    \item $\cD(X_{i \in [-n, n]})=\{X_0, X_{i-1}, X_i, X_{i+1}\}$, \\ 
    $\cD(X_{-(n + 1)}) = \{ X_0, X_{-(n + 1)}, X_{-n} \}$, \\
    $\cD(X_{n + 1}) = \{ X_0, X_{n} X_{n + 1} \}$;
\end{itemize}

\noindent
and equations $\cF^{\cA}$:
\begin{flalign*}
%&\cF_{X_0} ( X_0 = \langle q, \gamma, 0 \rangle ) = \{ \langle q', \gamma', d \rangle \mid \langle q', \gamma', d \rangle \in \deltaFinal(q, \gamma) \} &&\\
%
% &\cF_{X_0} ( X_0 = \langle q, \gamma, -1 \rangle, X_{-1} = \gamma' ) = &&\\
% &\qquad \{ \langle q', \gamma'', d \rangle \mid \langle q', \gamma'', d \rangle \in \deltaFinal(q, \gamma') \} &&\\
% %
% &\cF_{X_{i \in [-n, n + 1] \setminus \{0\}}} ( X_{0} = \langle q, \gamma, -1 \rangle, X_{i-1} =  \gamma' ) = \gamma' &&\\
% %
% &\cF_{X_0} ( X_0 = \langle q, \gamma, 1 \rangle, X_{1} = \gamma' ) = &&\\
% &\qquad \{ \langle q', \gamma'', d \rangle \mid \langle q', \gamma'', d \rangle \in \deltaFinal(q, \gamma') \} &&\\
% %
% &\cF_{X_{i \in [-(n + 1), n] \setminus \{0\}}} ( X_{0} = \langle q, \gamma, 1 \rangle, X_{i+1} = \gamma' ) = \gamma' && \\
%
&\cF_{X_0} ( X_0 = \langle q, \gamma, d \rangle, X_{d} = \gamma' ) = &&\\
&\qquad \{ \langle q', \gamma'', d' \rangle \mid \langle q', \gamma'', d' \rangle \in \deltaFinal(q, \gamma') \} &&\\
&\cF_{X_{i \in [-(n + 1), n] \setminus \{0\}}} ( X_{0} = \langle q, \gamma, d \rangle, X_{i+d} = \gamma' ) = \gamma' && 
\end{flalign*}
\end{definition}

In the initial configuration $\vec{v_0}$, $X_0 = \langle q_0, \leftmarker, 0 \rangle$, $X_{n + 1} = \rightmarker$, $X_{i \in [1, k]} = \sigma_i$, where $\sigma_1, \ldots, \sigma_k$ is the input, and $X_{i} = \#$ for $-n \leq i < 0$ and $k+1 \leq i \leq n$.

Informally, the first equation updates the value of $X_0$ and the second equation updates the value of all other variables $X_{i \in [-n, n + 1] \setminus \{0\}}$. 
%For both equations, there are three disjoint cases which depend on $d$: when the head did not move at the last transition $d = 0$, when the head moved left $d = -1$, and when the head moved right $d = 1$. 

For $X_0$, if the head did not move at $t - 1$ (this includes the initial configuration), the value of $X_0$ is updated with the LBA transition $(q', \gamma'', d') = \delta^F(q, \gamma')$ at $t$, where $\gamma' = \gamma$.
If the head moved left at the last transition, then $X_0$ is updated with $(q', \gamma'', d') = \delta^F(q, \gamma')$, where $\gamma'$ is the symbol in $X_{- 1}$ (i.e., the symbol currently under the head).
If the head moved right at the last transition, then $X_0$ is updated with $(q', \gamma'', d') = \delta^F(q, \gamma')$, where $\gamma'$ is the symbol in $X_{1}$.

For $X_{i \in [-n, n + 1] \setminus \{0\}}$, if the head did not move the value of $X_i$ is updated with $\gamma' = \gamma$, i.e., remains the same.
If the head moved left, the value of $X_i$ is updated  to be the value in $X_{i - 1}$, i.e., it shifts the representation of the LBA tape in the causal model to the right. (To simplify the presentation, we abuse notation and assume that in the case where $X_{i - 1} = X_0$, $\gamma'$ is the second element of the tuple in $X_0$.)
If the head moved right, the value of $X_i$ is updated to be the value in $X_{i + 1}$. 

We say that an LBA causal calculator \textit{accepts} a string $s$ iff the computation tree $\cC^{\cM^{\cA}}$ for the input $s$ contains a sequence of valid transitions $\vec{v_0} \twoheadrightarrow_{\cM^{\cal A}} \dots \twoheadrightarrow_{\cM^{\cal A}}\vec{v_m}$ from the root node $\vec{v_0}$ to some accepting configuration $\vec{v_m}$ such that $\vec{v_{i}} \in \ConfCM, 0 \leq i \leq m$. 
A configuration $\vec{v_m} = X_{-(n+1)}\ldots$ $X_{0} \ldots X_{n+1}$ is accepting iff $X_0 = \langle q^F \in F, \gamma, 0 \rangle$, i.e., the state is final.
%A configuration $\vec{v_m} = X_{-(n+1)} \ldots X_{0} \ldots X_{n+1}$ is accepting iff $X_0 = \langle q^F \in F, \rightmarker, 0 \rangle$, i.e., where the state is final and the head is above the right endmarker.

%We prove that for any LBA and any initial configuration of the LBA, the corresponding causal calculator matches computations on each branch step by step. By `matching at step $t$' we mean that the state of the LBA at step $t$ is the same as the state component of $X_0$ at step $t$ and for any cell $i \in \{0,\ldots,n+1\}$, the cell contains symbol $\gamma$ iff the symbol in $X_{c^t(i)}$ is $\gamma$. This will give us the following theorem:

\begin{theorem}\label{thm:LBA}
An LBA $\cal A$ accepts a string if, and only if, it is accepted by the corresponding causal calculator $\cM^{{\cal A}}$.
\end{theorem}
\begin{proof}
In order to prove this we show that there is a bisimulation relation between the nodes of the computation tree of $\cal A$, $\cT$, and of computation tree of $\cM^{{\cal A}}$, $\cC$, where bisimilar nodes have the same states and tape contents. If $\lambda$ and $\vec{v}$ are bisimilar, and $\lambda \twoheadrightarrow_{\cal A} \lambda'$, then there is a transition
$\vec{v} \twoheadrightarrow_{\cM^{\cal A}} \vec{v'}$ such that 
$\lambda'$ and $\vec{v'}$ are bisimilar, and vice versa. This will imply that a configuration with an accepting state is reachable in one tree if and only if it is reachable in another.

Since the variables of $\cM^{\cal A}$ correspond to different tape cells at different steps, the notion of bisimulation is non-trivial and depends on the history of the computation so far. We will say that a transition 
$\lambda \twoheadrightarrow_{\cal A} \lambda'$ is by $(q,\gamma,q', \gamma', d)$ if this was the element of $\delta^F$ used for this transition, and label this edge in $\cT$ by $d$. Similarly, we say that a transition $\vec{v} \twoheadrightarrow_{\cM^{\cal A}} \vec{v'}$ is by $(q,\gamma,q', \gamma', d)$ if the value of $X_0$ in $\vec{v}$ is $(q,\gamma,-)$ and the value of $X_0$ in $\vec{v'}$ is $(q',\gamma',d)$, and we label this edge in $\cC$ by $d$. 

Finally, we say that two nodes $\lambda$ and $\vec{v}$ at step $t$ are \emph{$d_1,\ldots, d_t$ bisimilar} if (i) the path from the root node to each of them is labelled $d_1,\ldots, d_t$, (ii) the state in $\lambda$ is the same as  the state component in the value of $X_0$ in $\vec{v}$, (iii) the symbol in $C_k$ in $\lambda$ is the same as the symbol in $X_{k-o_t}$ (the symbol in $X_0 = (-,\gamma,-)$ is $\gamma$), where $o_t = \Sigma^{i=t}_{i=0} d_i - d_t$. Observe that $j=\Sigma^{i=t}_{i=0} d_i$ is the position of the head in $\lambda$. 
(Note that the offset $o_t$ can also be calculated using the index of the variable containing $\leftmarker$, but this more general approach also works for non-deterministic TMs.)

Clearly, $\lambda_0$ and $\vec{v_0}$ are $\epsilon$-bisimilar, where $\epsilon$ is an empty path. 

For the inductive step, assume that $\lambda$ and $\vec{v}$ at step $t-1$ are $d_1,\ldots, d_{t-1}$ bisimilar. We need to show that \\
$(\cT \rightarrow \cC)$ if there is a transition from $\lambda$ to $\lambda'$ by $(q,\gamma,q', \gamma', d)$, then there is a transition from $\vec{v}$
to $\vec{v'}$ by $(q,\gamma,q', \gamma', d)$, and $\lambda'$ and $\vec{v'}$ are $d_1,\ldots, d_{t-1},d$ bisimilar;\\
$(\cC \rightarrow \cT)$  if there is a transition from $\vec{v}$ to $\vec{v'}$ by $(q,\gamma,q', \gamma', d)$, then there is a transition from $\lambda$
to $\lambda'$ by $(q,\gamma,q', \gamma', d)$, and $\lambda'$ and $\vec{v'}$ are $d_1,\ldots, d_{t-1},d$ bisimilar.

We prove 
%the back and forth conditions 
those conditions for $d=1$. The proofs for $d=-1$ and $d=0$ are similar.\\
$(\cT \rightarrow \cC)$ Let $\lambda = C_0\ldots q C_j\ldots C_{n+1}$, that is, the head is over $C_j$ and the state is $q$. By the inductive hypothesis, 
$X_0=(q,C_{j-d_{t-1}},d_{t-1})$, and the symbol that the head is reading is in $X_{j-o_{t-1}}$. Let the transition from $\lambda$ be by some $(q,\gamma,q',\gamma',1)$ where $\gamma$ is in $C_j$. The resulting $\lambda' = C_0\ldots C_j=\gamma' q'C_{j+1}\ldots C_{n+1}$.  The successor of $\vec{v}$ by $(q,\gamma,q',\gamma',1)$ is (i) reachable from the root by the same sequence $d_1,\ldots, d_{t-1},1$ as $\lambda'$,
(ii) has the same state $q'$ due to the equation for $X_0$, (iii)  the symbol that was held in $X_i$ is now held in $X_{i-1}$ because of the equation for $X_i$, which maintains the correspondence with LBA cells since the offset has increased by $d=1$.\\
$(\cC \rightarrow \cT)$ Let $X_0$ in $\vec{v}$ have value $(q,\chi,d_{t-1})$.
Assume that transition in $\cM^{\cal A}$ is to $\vec{v'}$ by $(q,\gamma,q',\gamma',1)$ where $\gamma$ is the symbol in $X_{d_{t-1}}$ (which by the inductive hypothesis is the symbol the head is reading in $\lambda$). Since $\vec{v}$ and $\lambda$ have the same state $q$ and the head position in $\lambda$ is over $\gamma$, ${\cal A}$ can make a transition to 
$\lambda' = C_0\ldots C_j=\gamma' q' C_{j+1}\ldots C_{n+1}$. The argument that $\vec{v'}$ is $d_1,\ldots, d_{t-1},1$ bisimilar to $\lambda'$ is the same as above.
\end{proof}

%\mg{$\delta \to \delta^F$?} done

%\input{X-deterministic-turing-machines.tex}
\section{Encoding Turing Machines}\label{sec:TM}

In this section we show that TSEMs with a countable set of variables $\cV$ are Turing complete. For a simpler construction, we consider the case of \emph{deterministic} Turing Machines, since they are known to be as powerful as non-deterministic ones. We begin by recalling the definition of a Turing Machine \cite{Hopcroft/Ullman:79a}.

\begin{definition}[Turing machine]\label{def:tm} A Turing Machine (TM) is a tuple $T = (Q, q_0, F, \Sigma, \Gamma, \delta)$, where
\begin{itemize}
    \item $Q$ is a finite set of states with $q_0\in Q$ the initial state;
    \item $F\subset Q$ is the set of final (or accepting) states;
    \item $\Sigma$ is a finite set of input symbols;
    \item $\Gamma = \Sigma \cup \{\# \}$ is the finite set of tape symbols, where $\# \not \in \Sigma$ is the blank symbol;
    \item $\delta: Q\setminus F\times \Gamma \longrightarrow Q\times\Gamma \times \{-1, 1\}$ is a function mapping a state $q\in (Q\setminus F)$ and a symbol $\gamma\in \Gamma$ to a tuple $\delta(q, \gamma) = (q', \alpha, d)$, where $q'$ is the next state, $\alpha\in \Gamma$ is a symbol that must be written in the cell and $d\in \{-1, 1\}$ is the direction in which the head moves ($-1$ indicates the head moves to the left). 
\end{itemize}
\end{definition}
\noindent
To simplify notation we will write $d=Dir(q, \gamma)$ (for direction), $\sigma = Sym(q, \gamma)$ (for symbol) and $q' = Out(q, \gamma)$ (for outcome).

We assume that all cells on the tape are indexed with integers relative to the head. So, the current cell always has index $0$, its right-side cells have positive indexes, and left-side cells have negative ones. We represent a \emph{configuration} of a TM as %a string 
$$\lambda = C_{-i}\dots C_{-2}C_{-1}qC_0C_1C_2\dots C_{j}$$ where $q$ is the state of the TM, the head is (always) scanning a cell indexed 0, and each $C_k$ represents a symbol written in the $k$-th cell of the tape. Note that only blank symbol $\#$ may occur on the tape on left of $-i$'s and on the right of $j$'s cells. 

\begin{definition} Let $Conf^{TM}(T)$ be the set of all possible configurations of a TM $T$. Then, a run of $T$ on input $I$ is a sequence $\lambda_0 \twoheadrightarrow_{T} \lambda_1 \twoheadrightarrow_{T}  \dots$ of elements of $Conf^{TM}(T)$, where $\lambda_0$ is the initial configuration of $T$ for input $I$, and each $\lambda_i$ yields $\lambda_{i+1}$ in one step. 
\end{definition}

A TM $T$ \emph{accepts} a string $s$ iff the run $\lambda_s \twoheadrightarrow_{T} \lambda_1 \twoheadrightarrow_{T}  \dots \twoheadrightarrow_{T} \lambda'$ of $T$ starting in the initial configuration $\lambda_s$ terminates in an accepting configuration $\lambda' =  C_{-i}\dots C_{-2}C_{-1}q^FC_0C_1C_2\dots C_{j}$ where the TM is in a final state  (i.e., $q^F \in F$).

Now we will construct a causal calculator $\cM^T$ simulating a TM $T$. Since a TM has an infinite tape, a causal calculator has to work with an input of arbitrary length. So, in contrast to Definition \ref{def:SEM}, we assume that $\cV$ contains countably many variables, however ranges and domains of all variables remain finite. The extension to infinitely many variables does not create any technical difficulties and the definitions of a computation $\cC$ (Definition \ref{def:computation}) remains unchanged. Causal models with infinitely many variables have previously been studied in \cite{Halpern2022NonRecurive}. 

\begin{definition}\label{def:canonicalCC}
    Given a TM $T = (Q, q_0, F, \Sigma, \Gamma, \delta)$, a TM causal calculator is a tuple $\cM^T = (\cS^T, \cF^T)$, with a signature $\cS^T = (\cV, \cR, \cD)$, where
\begin{itemize}
    \item $\cV = \{S\} \cup \{X_i\mid i\in\mathbb{Z}\}$;
    \item $\cR(X_i) = \Gamma$ and $\cR(S) = Q$;
    \item $\cD(S) = \{S, X_0\}$\\ 
    $\cD(X_i)=\{X_0, S, X_{i-1}, X_i, X_{i+1}\}$ for $i\in\mathbb{Z}$;
\end{itemize}

Structural equations $\cF$ are defined as follows. 

 $\cF_S(S=q, X_0=\gamma)= \begin{cases} q', \text{ if } q'\in Out(q, \gamma)\\
        q, \text{ if }  q\in F
    \end{cases}$

\noindent For $\vec{v}^{|\cD(X_{i})}=(X_0=\alpha, X_{i-1}=\beta, X_i=\gamma, X_{i+1}=\delta, S=q)$:

    $\cF_{X_i\notin \{-1, 1\}}(\vec{v}^{|\cD(X_{i})}) = \begin{cases} 
    \beta, \text{ if } -1= Dir(q, \alpha)\\
    \delta, \text{ if } 1= Dir(q, \alpha)\\
    \gamma, \text{ if } q\in F 
    \end{cases}$ 
        
   \noindent  For  $\vec{v}^{|\cD(_{X_{-1}})}=(X_0=\alpha, X_{-1}=\beta, X_{-2}=\gamma, S=q)$:
   
    $\cF_{X_{-1}}(\vec{v}^{|\cD(X_{-1})}) = \begin{cases} 
    \gamma, \text{ if } -1= Dir(q, \alpha)\\
    Sym(q, \alpha), \text{ if } 1= Dir(q, \alpha) \\
    \beta, \text{ if } q\in F 
    \end{cases}$ 

\noindent For $\vec{v}^{|\cD(X_{1})}=(X_0=\alpha, X_1=\beta, X_{2}=\gamma, S=q)$:
                                                                   
    $\cF_{X_{1}}(\vec{v}^{|\cD(X_{1})}) = \begin{cases} 
    Sym(q, \alpha), \text{ if } -1= Dir(q, \alpha)\\
    \gamma, \text{ if } 1= Dir(q, \alpha)\\
    \beta, \text{ if } q\in F
    \end{cases}$ 

\end{definition}

Variable $S$ represents the \textit{state} of the TM, and $X_i$'s represent corresponding cells on the tape. So, the range of $S$ is the set of states $Q$ of the TM, and the range of each $X_i$ is the tape alphabet $\Gamma$. Finally, domains $\cD$ demonstrate which variables are causally related. Thus, the value of the current state $S$ depends on the previous value of itself and a current symbol $X_0$. And the value of any variable $X_i$ depends on itself, $X_0, S$, and its left- and right-most cells $X_{i-1}$ and $ X_{i+1}$.  %Signature $\cS^T$ generates the dependency graph presented in Figure \ref{fig:causalTM}.

The structural equations are defined as follows. For $S$, if the previous value $q_1$ of $S$ is a final state of the TM ($q_1\in F$), then the current value of $S$ remains $q_1$. Otherwise, $S$ gets value $q_2$, which is the result of the transition  $\delta(q_1, \alpha)$ in the TM. Next, for any cell variable $X_i$, except $X_1$ and $X_{-1}$, and any possible input $\vec{v}^{|\cD(X_{i})}=(X_0=\alpha, X_{i-1}=\beta, X_i=\gamma, X_{i+1}=\delta, S=q)$ of $\cF_{X_i}$, it is defined as follows. In case the TM's head moves to the left ($-1= Dir(q, \alpha)$), $\cF_{X_i}$ gets the value of $X_{i-1}$; if it moves to the right ($1= Dir(q, \alpha)$) then $\cF_{X_i}$ gets the value of $X_{i+1}$; and if the state $q$ is final, $\cF_{X_i}$ inherits its own value from the previous step. This allows us to simulate the movements of the head in TM taking into account that the numeration of the tape cells is attached to the head. Thus, at any step of the computation $\cC$, the value of $X_i$ corresponds to the current symbol written in $i$'s cell on the tape of the TM. Finally, $X_{-1}$ and $X_1$ are two special cases, because TM rewrites the symbol in the current cell (indexed 0) with a new symbol $Sym(q, \alpha)$ on each step, which propagates to either $X_{-1}$ or $X_1$ depending on the direction of the next movement. Thus, for $X_{-1}$ and input $\vec{v}^{|\cD(X_{-1})}=(X_0=\alpha, X_{-1}=\beta, X_{-2}=\gamma, S=q)$, if the head moves left in $\delta(q, \alpha)$, then $\cF_{X_{-1}}$ gets the value of $X_{-2}$; if the head moves right in $\delta(q, \alpha)$, then $\cF_{X_{-1}}$ gets a symbol $Sym(q, \alpha)$ that has been written on the tape before the movement; and if $q$ is the final state, the value of $X_{-1}$ remains the same. The case for $X_1$ is symmetric.  

Now it is straightforward to define a translation between the configurations of $T$ and $\cM^T$. Given a configuration $\lambda = C_{-i}\dots C_{-2}C_{-1}qC_0C_1C_2\dots C_{j}$ of $T$, the corresponding configuration of $\cM$ is  $\vec{v}_{\lambda} = (S=q) \cup \{(X_k=C_k)\mid -i \leq k \leq j\} \cup \{(X_n=\#)\mid n\in \mathbb{Z}\setminus [-i, j] \}$. 

We say that a TM causal calculator \textit{accepts} a string $s$ iff the computation $\cC^{\cM^{T}}$ for the input $s$ contains a sequence of valid transitions $\vec{v_0} \twoheadrightarrow_{\cM^{T}} \dots \twoheadrightarrow_{\cM^{T}}\vec{v_m}$ from the initial configuration $\vec{v_0}$ to some accepting configuration $\vec{v_m}$ such that $\vec{v_{i}} \in \ConfCM, 0 \leq i \leq m$. 
A configuration $\vec{v_m}$ is accepting iff $(S = q^F)\in \vec{v_m}$, i.e., it contains the final state.

We want to show that for any initial configuration $\lambda$ of $T$ we can take a corresponding initial configuration $\vec{v} = \tau(\lambda)$ of $\cM^T$ and generate a computation $\cC$ which will simulate a computation of $T$ in a step-wise manner. 

\begin{lemma} For any $\lambda_1, \lambda_2\in Conf^{TM}$, it holds that $$\lambda_1\twoheadrightarrow_{T}\lambda_2,  \text{ iff } \vec{v}_{\lambda_1}\twoheadrightarrow_{\cM_T}\vec{v}_{\lambda_2}$$
\end{lemma}
\begin{proof}
    Given a configuration $$\lambda_1 = C_{-i}\dots C_{-2}C_{-1}qC_0C_1C_2\dots C_{j}$$ of the TM, there are three possible cases for the $$\lambda_2 = C_{-i}'\dots C_{-2}'C_{-1}'q'C_0'C_1'C_2'\dots C_{j}'$$
    
\textbf{Case} $q\in F$\footnote{Here for simplicity of the proof we assume that for any final configuration $\lambda$, $\lambda\twoheadrightarrow_{T}\lambda$.}. If the current state in $\lambda_1$ is final, then $\lambda_2=\lambda_1$. The construction of $\cM^T$ also guarantees that for any $Y\in \cV$ and any $\vec{v}$, s.t. $(S=q^F)\in \vec{v}$ it holds that $\cF_Y(\vec{v}^{|\cD(Y)})=\vec{v}^{|Y}$. So, $\vec{v}_{\lambda_1}\twoheadrightarrow_{\cM_T}\vec{v}_{\lambda_1}$, and thus $\vec{v}_{\lambda_1}\twoheadrightarrow_{\cM_T}\vec{v}_{\lambda_2}$. 

\textbf{Case} $-1= Dir(q, C_0)$. In this case the head moves to the left after writing symbol $\zeta$: $\delta(q, C_0)=(q', \zeta, -1)$. The construction of $\cM^T$ guarantees that (a) $\cF_S(\vec{v}_{\lambda_1}^{|\cD(S)}) = q'$ iff $q'=Out(q, C_0)$, (b) $\cF_{X_1}(\vec{v}_{\lambda_1}^{|\cD(X_1)}) = \zeta$ iff $\zeta = Sym(q, C_0)$ (i.e. iff $C_1'=\zeta$), and (c) for $i\neq 1$, $\cF_{X_i}(\vec{v}_{\lambda_1}^{|\cD(X_i)})=\beta$ iff $\vec{v}_{\lambda_1}^{|X_{i-1}}=\beta$ (i.e. iff $C_{i-1}=\beta$). Thus, $\lambda_1\twoheadrightarrow_{T}\lambda_2$ iff $\vec{v}_{\lambda_1}\twoheadrightarrow_{\cM_T}\vec{v}_{\lambda_2}$. 

\textbf{Case} $1= Dir(q, X_0)$. The argument is symmetric to the previous case.
\end{proof}

The next result follows immediately. 

\begin{theorem}\label{thm:turing_completeness}
    A TM accepts a string if, and only if, it is accepted by the corresponding causal calculator. 
\end{theorem}

As we mentioned before, it is a well-known fact that any non-deterministic TM $T^*$ can be simulated by a deterministic TM $T$ \footnote{Recall that in a non-deterministic TM $T^*$  the transition function $\delta$ is replaced with a relation 
$\delta^*\subseteq \bigl(Q\setminus F\times \Gamma\bigr) \times \bigl(Q\times\Gamma \times \{-1, 1\}\bigr)$}. So, Theorem \ref{thm:turing_completeness} implies that TSEMs can also simulate $T^*$ by simulating the corresponding $T$. However, such a procedure exponentially increases the computation time of $T$ (and thus of $\cM^T$) with respect to $T^*$. But since in general case TSEMs admit non-deterministic structural equations, they are capable of simulating an NTM in a stepwise manner, with no additional computing time. This can be demonstrated using a slightly modified construction of the LBA causal calculator from Definition \ref{def:CCLBA} that does not contain tape bounds. 

\begin{theorem}\label{thm:NTMCC}
An NTM accepts a string if, and only if, it is accepted by the corresponding causal calculator.
\end{theorem}
\begin{proof}
    The detailed proof is presented in Technical Appendix.
\end{proof}

\section{Potential Applications}\label{sec:interventions}

The key feature of causal models is that they are specifically designed for counterfactual reasoning. Causal models contain enough information to represent not only the actual sequence of events that occurred, but also hypothetical counterfactual alternatives. This capability aligns closely with theoretical concerns in computability theory, when we are interested in reasoning not only about what happened in the actual computation, but also what would have happened if a %particular transition had not occurred or if a 
different computational path had been taken. TSEMs provide sufficient expressive power to analyse both the actual computation and its counterfactual variations enabling new forms of reasoning about temporally extended processes by means of the 
%two types of 
interventions introduced in Section \ref{sec:temporalmodels}. 

Interventions allow us to alter value of some variable(s) 
%or change a specific structural equation(s) 
at a given time step and observe how this might change the overall outcome of a computation. 
For example, interventions allow us to ask questions such as ``Which part of the input had causal influence on the output?" or ``Had the random bit at step $t$ been flipped, would the output remain unchanged?" to analyse the internal causal structure that governs a computing device to arrive at a given result.

One possible application is the analysis of the causal influence of alternating bits of input, because 
%value 
interventions provide a straightforward way to explore how local changes in the input propagate through the computation. As a simple example, consider a Turing machine that accepts binary strings if the input contains an alternating pattern of 0s and 1s, e.g.``$0101\dots$" or ``$1010\dots$". A causal analysis of such a computation may reveal not only when the machine accepts or rejects the input, but also how individual bits contribute to the machine's decision. Suppose the machine halts and accepts the input ``0101'', which may be represented as the initial assignment of a canonical causal calculator (Definition \ref{def:canonicalCC}) $\vec{v_0}=(X_0=0,X_1=1,X_3=0,X_5=1$). Then, by systematically intervening on each bit, changing 0 to 1 or vice versa (in our notation $do(X_i^0\gets 1)$), we can identify which bits are causally necessary for acceptance. For example, the intervention $do(X_1^0\gets 0)$ disrupts the alternation pattern and causes the machine to reject, indicating that this bit $X_1$ has a critical causal link to the outcome of the computation. Conversely, if altering a bit leaves the outcome unchanged, we can classify it as causally inert or redundant. 

At the same time, we may be interested in intervening not only on the input bits (intervening at time step $t=0$), but also in analysing the role of some random bit alternation at an arbitrary step $t'$. This allows reasoning about computations with corrupted or faulty memory, for example. One potential application of such counterfactual reasoning is  the analysis of so-called soft errors \cite{soft_error}. A soft error is a type of a computer system fault that alters data or logic states without causing permanent damage to the hardware.  It is typically caused by external factors like cosmic rays or radioactive decay, which can flip a bit in memory or logic circuits, leading to incorrect computational behaviour. A recent example of this kind of error is the Airbus  incident, where intense solar radiation was claimed to corrupt data critical to the functioning of flight controls, which resulted in a recall of 6000 aircraft \cite{reuters}.  
By systematically intervening on different values 
%or structural equations 
in arbitrary steps of the computation and exploring the causal consequences of these changes, we can study the fault tolerance of a system and develop methods to mitigate the effects of soft errors in critical computational applications.

We conclude this section by commenting on the decision problems involved, namely establishing causality.
In Definition \ref{def:intervention}, as time indexes are fixed for all variables, the decision problem of whether some variable values are the cause of other variable values \emph{at a fixed step} is decidable for both LBA and TM causal calculators. However, if we would like to represent the outcome in a richer language, for example using an LTL temporal operator eventually $\mathsf{F}$, e.g. $\mathsf{F}(\vec{Y}=\vec{y})$ instead of $\vec{Y}^{\vec{t}}=\vec{y}$, then 
%both problems 
the problem becomes undecidable for the case of TM causal calculator, because the formula $\mathsf{F} (S=q^F)$ basically encodes the halting problem for Definition \ref{def:canonicalCC}.
\section{Discussion and Future Work} \label{sec:conclusion}

In 
%classical computability theory 
reasoning about computation, we are usually interested in answering questions like whether a given function is computable, how long it takes to compute, or how much resources it requires. However, we typically lack a formal framework for asking why a computation led to a specific outcome, or what would have happened if a particular input or step had been different. 

In this paper, we used temporal causal models to address this gap and to analyse the internal dynamics of computations in a causally meaningful way. We  
introduced a temporal causal framework with non-deterministic structural equations and defined interventions for this setting. We then proved that TSEMs can simulate Linear Bounded Automata and TSEMs with countably many variables are Turing complete. Finally, we discussed potential applications of interventions for counterfactual reasoning about computations. We believe our results can enrich the field of formal verification with TSEM-based techniques for reasoning about safety, fault tolerance and explainability of computational systems. 

In future work we would like to explore other notions of interventions that make sense for computations. For example, we could consider 
atomic \emph{structure} interventions, $do(Y^n(\vec{X}=\vec{x})\gets y$), where $\vec{X}=\cD(Y), \vec{x}\in\cR(\vec{X})$, says that at time $n$ the structural equation $\cF_Y$ for $Y$ must be rewritten so that $\cF_X(\vec{X}=\vec{x}) = y$. As before, a structure intervention $do(\vec{Y}^{\vec{n}}(\vec{X}=\vec{x})\gets \vec{y})$ is a vector of atomic interventions. 
%
% This definition of interventions is similar to the one used in \cite{Gladyshev_Alechina_Dastani_Doder_Logan_2025}. 

% \mg{we did not have structure intervention in [Gladyshev et al., 2025a]}

\begin{definition}[Structure Intervention] Given a TSEM $\cM$ and a structure intervention $I=do(\vec{Y}^{\vec{n}}(\vec{X}=\vec{x})\gets \vec{y})$, let dynamic structural equations $\cF^i$ be defined as follows. For any $Y\in \cV, \vec{X}=\cD(Y)$ and $\vec{x}=\cR(\vec{X})$, $\cF^0_Y(\vec{X}=\vec{x})= y'$ if $Y^0(\vec{X}=\vec{x})\gets y' \in I$ and $\cF^0_Y(\vec{X}=\vec{x})= \cF(\vec{X}=\vec{x})$ otherwise. For any $i>0$,
$$\cF^i_Y(\vec{X}=\vec{x})= \begin{cases}
    y', \text{ if } Y^i(\vec{X}=\vec{x})\gets y' \in I\\
   \cF_Y^{i-1}(\vec{X}=\vec{x}), \text{ otherwise} 
\end{cases}$$
Finally, the computation is defined as 
$$\cC^{I}(i) = \prod\limits_{X\in \cV} \cF_X^{i-1}\bigl(\cC^{I}(i-1)^{|\cD(X)}\bigr) \text{ for $i>0$}$$ 
\end{definition}

Note that structure interventions are not expressible in terms of standard interventions. Consider the model $\cV=\{X\}$, $\cR(X)=\{0,1\}$ and $\cF_X (X=x)=1$ for any $x$. So, whatever the initial assignment of $X$ is, starting from step $i>0, (X=1)\in C(i)$. With a structure intervention, we can turn the equation into a constant function returning 0 instead of 1: $\cF_X (X=x)=0$ for any $x$. So, we can enforce that for $i>0, (X=0)\in C(i)$. Standard interventions allow changing the value of $X$ only on the steps specified in the intervention. So, we can enforce $X=0$ on an arbitrary long but \emph{finite} interval of the computation. Thus, an infinite sequence of $X=0$ cannot be achieved by changing the value of $X$ at a finite sequence of steps.
Structure interventions would enable us to talk about changing transition rules in a computing device rather than a particular value, enabling us to model e.g. software patches.

% For future work, it is interesting to consider causal representations of various restricted versions of Turing Machine. In particular, we find linear bounded automata (which are essentially Turing machines restricted to a finite portion of the tape) especially interesting objects to study, because they can be represented with a finite temporal causal model, whose model-checking problem is was studied in \cite{Gladyshev_Alechina_Dastani_Doder_Logan_2025}.

%Another direction for future work is extending the setting to allow richer temporal languages. In Definition \ref{def:intervention}, as time indexes are fixed for all variables the decision problem of whether some variable values are the cause of other variable values \emph{at a fixed step} is decidable for both LBA and TM causal calculators. However, if we would like to represent the outcome in a richer language, for example using an LTL temporal operator eventually $\mathsf{F}$, e.g. $\mathsf{F}(\vec{Y}=\vec{y})$ instead of $\vec{Y}^{\vec{t}}=\vec{y}$, then 
%both problems 
%the problem becomes undecidable for the case of TM causal calculator, because the formula $\mathsf{F} (S=q^F)$ basically encodes the halting problem for Definition \ref{def:canonicalCC}.
Formal reasoning about structure interventions will require a richer syntax (being able to say that a particular equation is the cause of a state of affairs). We can also enrich the syntax with temporal operators, to express that some value is reachable in principle.

Further variations of causality problems for TSEMs may also be developed besides the extension with a richer syntax. Firstly, for non-deterministic computations it may be important to check if some $C$ is the cause of some outcome $O$ on all or only some branches of the computation tree. Secondly, we presented the simplest form of the definition of a cause, called but-for cause, however several alternative approaches to actual causality have been developed, e.g., \cite{HalpernBook,Beckers2021,beckers2025actual}. The adaptation of these approaches to temporal settings and validating whether arguments against but-for definition hold for TSEM framework as well as extending existing formal verification techniques, e.g., \cite{Lorini2024,KR2025-59} to temporal setting are important directions for future research. 

Finally, the canonical causal calculators constructed in Sections \ref{sec:LBA}--\ref{sec:TM} have a specific structure since they are specifically designed to simulate LBA and TMs. However, other architectures of causal models (e.g. a more distributed structures with multiple `inner' variables) may be more suitable in some applications.  In this context, an important question to ask is whether two structurally different causal models execute the same computational behaviour under various interventions. We believe that recent work on the notions of causal model equivalence  \cite{Beckers2021eq,Gladyshev_Alechina_Dastani_Doder_Logan_2025} as well as on the notions of an \textit{abstraction} or a \textit{transformation} of one causal model into another \cite{rubenstein2017causal,Beckers_abstractions,willig2023do} can bring new formal techniques initially developed for causal reasoning to analysis of computation and enrich the field.

\bibliographystyle{named}
\bibliography{ijcai26}

%\clearpage
%% Appendix:
\clearpage
\section*{Technical Appendix}

\paragraph{Stepwise simulation of an NTM}
Here we demonstrate how to extend the construction of the LBA causal calculator from Section 3 to simulate a non-deterministic TM (NTM) is stepwise manner and prove \textbf{Theorem 3}. The overall argument is very similar to Section 3, except the numeration of the tape cells.  We start with a definition of NTM \cite{Hopcroft/Ullman:79a}.

\begin{definition} A non-deterministic Turing Machine  (NTM) is a tuple $T^* = (Q, q_0, F, \Sigma, \Gamma,\delta^*)$, in which all elements except $\delta^*$ are identical to TM (Definition 8), and the transition function is replaced with a relation 
$$\delta^*\subseteq \bigl(Q\setminus F\times \Gamma\bigr) \times \bigl(Q\times\Gamma \times \{-1, 0, 1\}\bigr)\footnote{Here for technical reasons we assume that $d=\{-1, 0, 1\}$ rather than $\{-1, 1\}$.  }$$
\end{definition}

%\mg{Does footnote 4 requires additional explanation? }
%\natasha{or why don't we just remove 0 from NTM? Is it because you don't want it to stop?}
%\mg{Mainly because in the initial step we will need to assign $X_0$ with some direction and each $C_k$ corresponds to $X_{k+d_t}$. In this case the direction of the previous movement would be non 0 and would require a special explanation why the correspondence between cells is different at step 0}

Following Section 4, we represent a \emph{configuration} of an NTM as a string 
$$\lambda = C_{-i}\dots C_{-2}C_{-1}qC_0C_{1}C_{2}\dots C_{j}$$ 
where $q$ is the state of the TM, the tape head is scanning a symbol indexed $i$, and each $C_l$, $-i \leq l \leq j$ represents a symbol written in the $l$-th cell of the tape. Note that only the blank symbol $\#$ may occur in cells to the left of $-i$ and to the right of $j$. We use $\ConfTM(T^*)$ to denote the set of all configurations of $T^*$.
%\mg{I changed the definition of a configuration to match Section 4}

\begin{definition} A run of an NTM $T^*$ on input $I$ is a tree $\cT$. The nodes of $\cT$ are elements of $\ConfTM(T^*)$ with a root node $\lambda_0$ representing the initial configuration of
$T^*$ for input $I$, and $\lambda \twoheadrightarrow_{T^*}\lambda' \in \cT$ iff $\lambda'$ is a possible result of transition from $\lambda$ according to $\delta^*$.
\end{definition}

Similarly to Section 3, to simplify the construction, we introduce a transition relation $\deltaFinal$ which is deterministic for final states of the NTM:
$$\deltaFinal(q, \gamma) = \begin{cases} 
    \{ \langle q', \gamma', d \rangle \mid \langle q, \gamma, q', \gamma', d \rangle \in \delta^*(q, \gamma) \} \text{ if } q \not\in F \\
    \{ \langle q, \gamma, 0 \rangle \}  \text{ if } q\in F 
    \end{cases}$$
% $$\deltaFinal(q, \gamma) = \begin{cases} 
%     \{ \langle q', \gamma', d \rangle \mid \\
%     \quad \langle q, \gamma, q', \gamma', d \rangle \in \delta(q, \gamma) \} \text{ if } q \not\in F \vee \gamma \not = \rightmarker \\
%     \{ \langle q, \gamma, 0 \rangle \}  \text{ if } q\in F \wedge \gamma = \rightmarker
%     \end{cases}$$
where $d \in \{-1, 0, 1\}$.

%Now we want to extend the result of Theorems \ref{thm:turing_completeness}--\ref{thm:DTIME} to the case of non-deterministic computations. For this purpose we extend TSEMs framework with non-deterministic structural equations discussed in \cite{beckers2025nondeterministic}. Then, we demonstrate that non-deterministic temporal causal models can simulate non-deterministic Turing Machines.   

% Of course, one might argue that given an NTM $T_1^*$ we may construct a DTM $T_2$ simulating $T_1^*$ and then simulate $T_2$ with a causal calculator $\cM^{T_2}$ constructed above. 

% \begin{definition} A \emph{non-deterministic} causal model $\cM^* = (\cS, \cF^*)$ is a model with non-deterministic structural equations defined as
% $$\mathcal{F^*}_Y:\prod\limits_{X \in \cD(Y)}\cR(X)\longrightarrow 2^{\cR(Y)}$$    
% \end{definition}

% \begin{definition} 
% A computation $\cC^*$ for $\cM^*$ with initial assignment $\vec{v_0}$ is a tree $\cT^{\cC}$, whose nodes are elements of $Conf^{CM}(\cM^*)$ with a root
% node $\vec{v_0}$, and $\vec{v_1} \twoheadrightarrow_{\cM^*} \vec{v_2}$ iff $\forall X\in\cV, \vec{v_2}^{|X}\in \cF_X(\vec{v_1})$.
% \end{definition}

An NTM \emph{accepts} a string $s$ iff the computation tree $\cT$ for the input $s$ contains a sequence of valid transitions $\lambda_s \twoheadrightarrow_{T^*} \dots \twoheadrightarrow_{T^*}\lambda'$ from the root node $\lambda_s$ to some accepting configuration $\lambda'$, such that $q^F\in \lambda'$. 

\begin{definition}\label{def:CC_NTM} Given a non-deterministic TM $T^* = (Q, q_0, F, \Sigma, \Gamma, \delta^*)$,  a non-deterministic canonical causal calculator is a tuple $\cM^{T^*} = (\cS^{T^*}, \cF^{T^*})$ with signature $\cS^{T^*} = (\cV, \cR, \cD)$ where
\begin{itemize}
    \item $\cV = \{X_i\mid i\in \mathbb{Z}\}$;
    \item $\cR(X_0) = Q \times \Gamma \times \{-1, 0, 1\}$, \\
    $\cR(X_{i \neq 0}) = \Gamma$;
    \item $\cD(X_{i})=\{X_0, X_{i-1}, X_{i+1}\}$;
\end{itemize}

\noindent
and equations $\cF^{T^*}$:
%\natasha{Brian is rewriting the equations for LBA to add the case when the head does not move, and the explanations, we can paste them here afterwards or just leave it as it is, unlikely people will read very carefully.}
\begin{flalign*}
% &\cF_{X_0} ( X_0 = \langle q, \gamma, 0 \rangle ) = \{ \langle q', \gamma', d \rangle \mid \langle q', \gamma', d \rangle \in \deltaFinal(q, \gamma) \} &&\\
% %
% &\cF_{X_0} ( X_0 = \langle q, \gamma, -1 \rangle, X_{-1} = \gamma' ) = &&\\
% &\qquad \{ \langle q', \gamma'', d \rangle \mid \langle q', \gamma'', d \rangle \in \deltaFinal(q, \gamma') \} &&\\
% %
% &\cF_{X_{i\neq 0}} ( X_{0} = \langle q, \gamma, -1 \rangle, X_{i-1} =  \gamma' ) = \gamma' &&\\
% %
% &\cF_{X_0} ( X_0 = \langle q, \gamma, 1 \rangle, X_{1} = \gamma' ) = &&\\
% &\qquad \{ \langle q', \gamma'', d \rangle \mid \langle q', \gamma'', d \rangle \in \deltaFinal(q, \gamma') \} &&\\
% %
% &\cF_{X_{i \neq 0}} ( X_{0} = \langle q, \gamma, 1 \rangle, X_{i+1} = \gamma' ) = \gamma' &&\\ 
%
&\cF_{X_0} ( X_0 = \langle q, \gamma, d \rangle, X_{d} = \gamma' ) = &&\\
&\qquad \{ \langle q', \gamma'', d' \rangle \mid \langle q', \gamma'', d' \rangle \in \deltaFinal(q, \gamma') \} &&\\
&\cF_{X_{i\neq 0}} ( X_{0} = \langle q, \gamma, d \rangle, X_{i+d} = \gamma' ) = \gamma' &&
\end{flalign*}
\end{definition}

This construction is similar to Definition 7, except the numeration on the tape and the absence of end-marker symbols since the NTM's tape is unbounded. So, in the initial configuration $\vec{v_0}$, $X_0 = \langle q_0, \alpha, 0 \rangle$, where $C_0=\alpha$, $X_{k\in [-i, j]\setminus \{0\}} = C_k$, and $X_{l} = \#$ for $l<-i$ and $l>j$.

For $X_0$, if the head did not move at $t - 1$ (this includes the initial configuration), the value of $X_0$ is updated with the NTM transition $(q', \gamma'', d') = \delta^F(q, \gamma')$ at $t$, where $\gamma' = \gamma$.
If the head moved left at the last transition, then $X_0$ is updated with $(q', \gamma'', d') = \delta^F(q, \gamma')$, where $\gamma'$ is the symbol in $X_{- 1}$ (i.e., the symbol currently under the head).
If the head moved right at the last transition, then $X_0$ is updated with $(q', \gamma'', d') = \delta^F(q, \gamma')$, where $\gamma'$ is the symbol in $X_{1}$.

For $X_{i\neq 0}$, if the head did not move the value of $X_i$ is updated with $\gamma' = \gamma$, i.e., remains the same.
If the head moved left, the value of $X_i$ is updated to be the value in $X_{i - 1}$ (To simplify the presentation, we abuse notation and assume that in the case where $X_{i - 1} = X_0$, $\gamma'$ is the second element of the tuple in $X_0$.)
If the head moved right, the value of $X_i$ is updated to be the value in $X_{i + 1}$.

We say that an NTM causal calculator \textit{accepts} a string $s$ iff its computation tree $\cC$ for the input $s$ contains a sequence of valid transitions $\vec{v_0} \twoheadrightarrow_{\cM^{T^*}} \dots \twoheadrightarrow_{\cM^{T^*}}\vec{v_m}$ from the root node $\vec{v_0}$ to some accepting configuration $\vec{v_m}$ such that $\vec{v_{i}} \in \ConfCM, 0 \leq i \leq m$. 
A configuration $\vec{v}$ is accepting iff $(X_0 = \langle q^F \in F, \gamma, 0 \rangle)\in \vec{v}$, i.e., if it contains the final state.
%A configuration $\vec{v_m} = X_{-(n+1)} \ldots X_{0} \ldots X_{n+1}$ is accepting iff $X_0 = \langle q^F \in F, \rightmarker, 0 \rangle$, i.e., where the state is final and the head is above the right endmarker.

%We prove that for any LBA and any initial configuration of the LBA, the corresponding causal calculator matches computations on each branch step by step. By `matching at step $t$' we mean that the state of the LBA at step $t$ is the same as the state component of $X_0$ at step $t$ and for any cell $i \in \{0,\ldots,n+1\}$, the cell contains symbol $\gamma$ iff the symbol in $X_{c^t(i)}$ is $\gamma$. This will give us the following theorem:

% \mg{change the Thm number}
% \begin{theorem}
% An NTM accepts a string if, and only if, it is accepted by the corresponding causal calculator.
% \end{theorem}
\paragraph{Theorem 3.} \textit{An NTM accepts a string if, and only if, it is accepted by the corresponding causal calculator.}
\begin{proof}
Following the proof of Theorem 1, we will show that there is a bisimulation relation between the nodes of the computation tree of $T^*$, $\cT$, and of computation tree of $\cM^{{T^*}}$, $\cC$. If $\lambda$ and $\vec{v}$ are bisimilar, and $\lambda \twoheadrightarrow_{\cal A} \lambda'$, then there is a transition
$\vec{v} \twoheadrightarrow_{\cM^{T^*}} \vec{v'}$ such that 
$\lambda'$ and $\vec{v'}$ are bisimilar, and vice versa. This will imply that a configuration with an accepting state is reachable in one tree if and only if it is reachable in another.

We say that a transition 
$\lambda \twoheadrightarrow_{T^*} \lambda'$ is by $(q,\gamma,q', \gamma', d)$ if this was the element of $\delta$ used for this transition, and label this edge in the tree by $d$. Similarly, we say that a transition $\vec{v} \twoheadrightarrow_{\cM^{T^*}} \vec{v'}$ is by $(q,\gamma,q', \gamma', d)$ if the value of $X_0$ in $\vec{v}$ is $(q,\gamma,-)$ and the value of $X_0$ in $\vec{v'}$ is $(q',\gamma',d)$, and we label this edge in the tree by $d$. 

Finally, we say that two nodes $\lambda$ and $\vec{v}$ at step $t$ are \emph{$d_1,\ldots, d_t$ bisimilar} if (i) the path from the root node to each of them is labelled $d_1,\ldots, d_t$, (ii) the state in $\lambda$ is the same as  the state component in the value of $X_0$ in $\vec{v}$, (iii) the symbol in $C_k$ in $\lambda$ is the same as the symbol in $X_{k+d_t}$ (the symbol in $X_0 = (-,\gamma,-)$ is $\gamma$).

%\mg{this is the main modification, correspondence of indexes is now simpler} 
%\natasha{I don't understand, t is just the step number? It is possible that on the first step the head goes right and the correspondence becomes %$C_k$ corresponds to $X_{k-1}$, e.g. $C_0$ with left marker corresponds to $X_{-1}$}
%\mg{Yes, it should be $d_t$, not $t$. So, after first step $C_0$ corresponds to $X_1$. If the next transition is again right, then again $C_0$ again corresponds to $X_1$. Note that the current head position of NTM is fixed to 0 and does not change unlike in the case of LBA, so after the first transition $C_0$ does not read the end-marker, but the first symbol}
%\natasha{No it is $X_{-1}$, not $X_1$. Something is wrong here. Shall we omit this proof to be on the safe side (and adjust the text in the paper accordingly?}

Clearly, $\lambda_0$ and $\vec{v_0}$ are $\epsilon$-bisimilar. 

For the inductive step, assume that $\lambda$ and $\vec{v}$ at step $t-1$ are $d_1,\ldots, d_{t-1}$ bisimilar. We need to show that \\
$(\cT \rightarrow \cC)$ if there is a transition from $\lambda$ to $\lambda'$ by $(q,\gamma,q', \gamma', d)$, then there is a transition from $\vec{v}$
to $\vec{v'}$ by $(q,\gamma,q', \gamma', d)$, and $\lambda'$ and $\vec{v'}$ are $d_1,\ldots, d_{t-1},d$ bisimilar;\\
$(\cC \rightarrow \cT)$  if there is a transition from $\vec{v}$ to $\vec{v'}$ by $(q,\gamma,q', \gamma', d)$, then there is a transition from $\lambda$
to $\lambda'$ by $(q,\gamma,q', \gamma', d)$, and $\lambda'$ and $\vec{v'}$ are $d_1,\ldots, d_{t-1},d$ bisimilar.

We prove the back and forth conditions for $d=1$. The proofs for $d=-1$ and $d=0$ are similar.\\
$(\cT \rightarrow \cC)$ Let $\lambda = C_{-i}\dots C_{-1}qC_0C_{1}\dots C_{j}$. By the inductive hypothesis, 
$X_0=(q,C_{-d_{t-1}},d_{t-1})$, and the symbol that the head is reading is in $X_{d_{t-1}}$. Let the transition from $\lambda$ be by some $(q,\gamma,q',\gamma',1)$ where $\gamma$ is in $C_0$. The resulting $\lambda' = C_{-i-1}\dots C_{-1}qC_0=\gamma'C_{1}\dots C_{j-1}$.  The successor of $\vec{v}$ by $(q,\gamma,q',\gamma',1)$ is (i) reachable from the root by the same sequence $d_1,\ldots, d_{t-1},1$ as $\lambda'$,
(ii) has the same state $q'$ due to the equation for $X_0$, (iii)  the symbol that was held in $X_i$ is now held in $X_{i-1}$ because of the equation for $X_i$, which maintains the correspondence with LBA cells since the offset has changed by $d=1$.\\
$(\cC \rightarrow \cT)$ Let $X_0$ in $\vec{v}$ have value $(q,\chi,d_{t-1})$.
Assume that transition in $\cM^{T^*}$ is to $\vec{v'}$ by $(q,\gamma,q',\gamma',1)$ where $\gamma$ is the symbol in $X_{d_{t-1}}$ (which by the inductive hypothesis is the symbol the head is reading in $\lambda$). Since $\vec{v}$ and $\lambda$ have the same state $q$ and the head position in $\lambda$ is over $\gamma$, ${T^*}$ can make a transition to 
$\lambda' = C_{-i}\ldots C_{-1} q'C_0=\gamma' C_{1}\ldots C_{j}$, where $C_k$ from $\lambda'$ is equal to $X_{k+1}$ from $\vec{v'}$.  The argument that $\vec{v'}$ is $d_1,\ldots, d_{t-1},1$ bisimilar to $\lambda'$ is the same as above.
\end{proof}

%\bsl{Monolithic encodng of an LBA}

\paragraph{Single variable encoding of an LBA   }

Recall that the construction of the LBA causal calculator in Section 3 uses twice as many variables as the number of LBA cells. However, one can use even simpler construction with a single variable, whose values represent all possible configurations of the LBA.

%We now show how to construct a causal calculator which mimics a given LBA. 
\begin{definition}[LBA Causal Calculator]\label{def:monolithic_CC}
    Given an LBA ${\cal A} = (Q, q_0, F, \Sigma, \Gamma, \delta)$ 
    with tape of length $n$, an \emph{LBA causal calculator} is a tuple $\cM^{{\cal A}} = (\cS^{\cal A}, \cF^{\cal A})$, with a signature $\cS^{\cal A} = (\cV, \cR, \cD)$, where
\begin{itemize}
    \item $\cV = \{V\}$ where $V$ represents a configuration $S H X_0 X_{1}X_{2}\ldots X_{n} X_{n+1}$
    \item $\cR(V) = Q \times \{0,\ldots,n+1\} \times \{\leftmarker\} \times \Gamma \times \ldots \Gamma \times \{\rightmarker\}$
    \item $\cD(V) = \{V\}$
\end{itemize}

The structural equations $\cF$ are defined as follows:\\ 
$\cF_V (q,j,\leftmarker, \gamma_1,\ldots,\gamma_j,\ldots,\gamma_n,\rightmarker)=$\\
$\{(q',j',\leftmarker, \gamma_1,\ldots,\gamma'_j,\ldots,\gamma_n,\rightmarker) \mid j' = j+d, (q,\gamma_j,q',\gamma'_j, d) \in \delta\}$ 
\end{definition}

Given a configuration $\lambda$ of the LBA $\cal A$, we denote a corresponding configuration of $\cM^{\cal A}$ as $(V=\lambda)$. For any initial configuration $\lambda_0$, let $\cT_{\cA}$ and $\cC_{\cM^{\cA}}$ be computation trees for the LBA $\cA$ and the causal calculator $\cM^A$ on inputs $\lambda_0$ and $\vec{v_0} = (V=\lambda_0)$ respectively. It is now straightforward to prove that the construction of $\cM^{\cal A}$ guarantees that $(V=\lambda)\twoheadrightarrow_{\cM^{\cA}} (V=\lambda') \in \cC_{\cM^{\cA}}$ iff $\lambda \twoheadrightarrow_{\cal A}\lambda' \in \cT_{\cA}$. Thus, $\cT_{\cA}$ is isomorphic to $\cC_{\cM^{\cA}}$.

However, this construction does not allow us to intervene on e.g. the symbol read by the head or on the next transition of the LBA.

%\mg{How can we justify that the construction in Section 3 is better?}

%It is now straightforward to prove the following theorem:

% \begin{theorem}\label{thm:LBA_simulation}
%    Any LBA can be simulated by a TSEM.
% \end{theorem}
% \begin{proof}
%     Given a configuration $\lambda$ of the LBA $\cal A$, we denote a corresponding configuration of $\cM^{\cal A}$ as $(V=\lambda)$. For any initial configuration $\lambda_0$, let $\cT_{\cA}$ and $\cT_{\cM^{\cA}}$ be computation trees for the LBA $\cA$ and the causal calculator $\cM^A$ on inputs $\lambda_0$ and $\vec{v_0} = (V=\lambda_0)$ respectively. The construction of $\cM^{\cal A}$ guarantees that $(V=\lambda)\twoheadrightarrow_{\cM^{\cA}} (V=\lambda') \in \cT_{\cM^{\cA}}$ iff $\lambda \twoheadrightarrow_{\cal A}\lambda' \in \cT_{\cA}$. Thus, $\cT_{\cA} = \cT_{\cM^{\cA}}$.
% \end{proof}

%\clearpage
%\input{CoverLetter}
\end{document}